\documentclass[preliminary]{eptcs}
\providecommand{\event}{QPL 2011}

\usepackage{graphicx}

\usepackage{proof}       
\usepackage{amssymb}
\usepackage{cmll}
\usepackage{stmaryrd}
\newenvironment{varitemize}
{
\begin{list}{\labelitemi}
{\setlength{\itemsep}{0pt}
 \setlength{\topsep}{0pt}
 \setlength{\parsep}{0pt}
 \setlength{\partopsep}{0pt}
 \setlength{\leftmargin}{15pt}
 \setlength{\rightmargin}{0pt}
 \setlength{\itemindent}{0pt}
 \setlength{\labelsep}{5pt}
 \setlength{\labelwidth}{10pt}
}}
{
 \end{list} 
}

\newcounter{numberone}

\newenvironment{varenumerate}
{
\begin{list}{\arabic{numberone}.}
{
  \usecounter{numberone}
  \setlength{\itemsep}{0pt}
  \setlength{\topsep}{0pt}
  \setlength{\parsep}{0pt}
  \setlength{\partopsep}{0pt}
  \setlength{\leftmargin}{15pt}
  \setlength{\rightmargin}{0pt}
  \setlength{\itemindent}{0pt}
  \setlength{\labelsep}{5pt}
  \setlength{\labelwidth}{15pt}
}}
{
\end{list} 
}

\newcommand{\QMLL}{\ensuremath{\textsf{QMLL}}}

\newcommand{\MLL}{\ensuremath{\textsf{MLL}}}

\newcommand{\CC}{\mathbb{C}}

\newcommand{\unone}{U}
\newcommand{\untwo}{V}
\newcommand{\mmid}{\; \; \mbox{\Large{$\mid$}}\;\;}

\newcommand{\typeone}{A}
\newcommand{\typetwo}{B}
\newcommand{\typethree}{C}
\newcommand{\typefour}{D}
\newcommand{\qtypeone}{Q}
\newcommand{\qtypetwo}{R}

\newcommand{\mtypeone}{F}
\newcommand{\mtypetwo}{G}

\newcommand{\atomone}{\alpha}
\newcommand{\conone}{C}
\newcommand{\contwo}{D}
\newcommand{\conthree}{E}
\newcommand{\pconone}{P}

\newcommand{\nconone}{N}

\newcommand{\tens}{\otimes}

\newcommand{\ctxone}{\Gamma}
\newcommand{\ctxtwo}{\Delta}
\newcommand{\ctxthree}{\Theta}

\newcommand{\pfone}{\pi}
\newcommand{\pftwo}{\rho}
\newcommand{\pfthree}{\sigma}
\newcommand{\pffour}{\xi}
\newcommand{\pffive}{\mu}
\newcommand{\axiom}{\mathsf{A}}
\newcommand{\cut}{\mathsf{C}}
\newtheorem{theorem}{Theorem}[section]   % Numbered within each section
\newtheorem{lemma}[theorem]{Lemma}
\newtheorem{proposition}[theorem]{Proposition}

\newenvironment{proof}{\begin{trivlist}
       \item[\hskip \labelsep {\bfseries Proof.}]}{\hfill $\Box$ \end{trivlist}}

\renewcommand{\aa}{\alpha}
\renewcommand{\b}{{}^{\bot}}

\newcommand{\qq}[2]{\qcnp{#1}{#2}}
\newcommand{\qqb}[2]{\bqcnp{#1}{#2}}

\newcommand{\bt}[1]{#1^{\bot}}

\newcommand{\qcc}{\boxdot}
\newcommand{\bqcc}{\diamonddot}
\newcommand{\qcn}[1]{\boxdot #1}
\newcommand{\bqcn}[1]{\diamonddot #1}
\newcommand{\qcnp}[2]{\boxdot^{#2}#1}
\newcommand{\bqcnp}[2]{\diamonddot^{#2}#1}

\newcommand{\jd}[1]{\vdash #1}
\newcommand{\jdp}[2]{#1:\;\vdash #2}
\newcommand{\natone}{n}

\newcommand{\unop}[1]{\mathbb{U}_{#1}}

\newcommand{\idmat}[1]{I_{#1}}
\newcommand{\qrule}[1]{\mathsf{QR}_{#1}}
\newcommand{\qqrule}[1]{\mathsf{MR}_{#1}}
\newcommand{\mqrule}[1]{\mathsf{RR}_{#1}}

\newcommand{\becomes}{\Longrightarrow}
\input{Qcircuit}
\usepackage{MnSymbol}

\newcommand{\qcone}{\mathbf{Q}}
\newcommand{\ec}{[\cdot]}
\newcommand{\stckone}{s}
\newcommand{\stcktwo}{r}
\newcommand{\stckthree}{q}
\newcommand{\regone}{Q}
\newcommand{\regtwo}{R}
\newcommand{\csb}[2]{#1[#2]}
\newcommand{\length}[1]{|#1|}
\newcommand{\depth}[1]{\partial(#1)}

\newcommand{\qiam}[1]{{M}_{#1}}
\newcommand{\sts}[1]{{Q}_{#1}}
\newcommand{\ists}[1]{{IQ}_{#1}}
\newcommand{\fsts}[1]{{FQ}_{#1}}

\newcommand{\trs}[1]{\longrightarrow_{#1}}

\newcommand{\stone}{S}
\newcommand{\sttwo}{T}
\newcommand{\stthree}{R}
\newcommand{\emstck}{\varepsilon}
\newcommand{\sem}[1]{\llbracket #1 \rrbracket}
\newcommand{\semp}[2]{\llbracket #1 \rrbracket_{#2}}

\newcommand{\IAM}{\textsf{IAM}}
\newcommand{\QIAM}{\textsf{QIAM}}

\title{On Multiplicative Linear Logic,\\ Modality and Quantum Circuits}
\author{
Ugo Dal Lago
\institute{Universit\`a di Bologna \& INRIA }
\email{dallago@cs.unibo.it}
\and
Claudia Faggian
  \institute{CNRS \& Universit\'e Denis-Diderot Paris 7}
\email{faggian@pps.jussieu.fr}
}
\def\titlerunning{On Multiplicative Linear Logic, Modality and Quantum Circuits}
\def\authorrunning{U. Dal Lago \& C. Faggian}

\begin{document}
\maketitle
\begin{abstract}
A logical system derived from linear logic and called \QMLL\ is 
introduced and shown able to capture all unitary quantum circuits. 
Conversely, any proof is shown to compute, through a concrete 
GoI interpretation, some quantum circuits. The system \QMLL, 
which enjoys cut-elimination, is obtained by endowing multiplicative linear 
logic with a quantum modality.
\end{abstract}

%%%%%%%%%%%%%%%%%%%%%%%%
\section{Introduction}
%%%%%%%%%%%%%%%%%%%%%%%%
It's more and more clear that strong relationships
exist between linear logic \cite{Girard:LL}  and quantum computation. This
seems to go well beyond the easy observation that the
intrinsic resource-consciousness of linear logic copes well with
the impossibility of cloning and erasing qubits.
There are several different research directions in which this 
interaction has recently started to manifest itself. We like to mention 
the following:
\begin{varitemize}
\item
  First of all, various lambda calculi  for quantum computation
  have been introduced in the last ten years~\cite{VanTonder, SelingerV06, SelingerV08, DLMZ}.
  The common denominator between the different proposals is precisely the
  use of linearity to control duplication and erasing: it is enforced either
  by typing or by structural constraints on the shape of lambda terms.
\item Coherence spaces (the semantics from which linear logic originated) have been revisited by
  Girard~\cite{Girard:QCS}, in an attempt to improve the understanding of 
  the relations between quantum and logic, and to relate coherence spaces 
  and quantum actions.  
\item
  Blute and Panangaden have recently shown how a simple calculus of Feynman's diagrams
  can give semantics to linear logic proof-nets~\cite{Blute_proofnets}.
\end{varitemize}
Even with all these recent advances, we still lack a  truly convincing correspondence between linear
logic as a proof theory and quantum computation as a computational model. In particular, 
a quantum analogue of the Curry-Howard correspondence has not been defined yet, and the 
known attempts (e.g.~\cite{Duncan}) have not any direct relationship with linearity in the sense of linear logic.

At a deeper level, a fundamental aspect of linear logic, which has not been exploited in a quantum setting
yet, is its rooting into a mathematical model based on operator algebras, through the so-called geometry of 
interaction~\cite{Girard87a,Girard88,DanosRegnier:PNHil} (GoI in the following).  
The GoI program is that of a dynamic interpretation of computation as a 
flow of information circulating around a net; this is at the heart of linear logic from its beginnings.
This flow of information can be formulated  both as a classical, token-based interactive machine~\cite{GonthierAbadiLevy} or 
as an algebra of bounded operators on the infinite dimension Hilbert space~\cite{Girard88}, which is 
the canonical state space for quantum computation models like quantum
Turing machines. We believe that this aspect is highly relevant to the logical approach to quantum computing, 
and has the potential for turning it into a powerful tool. Recently, GoI has been shown to be
able to give semantics to Selinger and Valiron's quantum lambda calculus~\cite{Hasuo11}.

In this paper, we describe  ongoing work about the relationships between quantum computation and 
linear logic and a first investigation on the underlying GoI. 
Some motivations and goals underlie and guide our investigation:
\begin{varitemize}
\item 
  we would like to get a model which is concrete, together with an efficient encoding; in our view, 
  GoI should be able to give a concrete syntax (as is the case of linear logic), 
  and should eventually be able to talk about the computational 
  complexity of the calculus. 
\item 
  we aim to have a proof-theoretical account of both quantum information and 
  computation (i.e. quantum data and algorithms).
\end{varitemize}
Specifically, in this paper we introduce a logical system, called \QMLL, which is 
obtained by endowing multiplicative linear logic with quantum modalities; we then investigate in detail 
the relations between \QMLL\ and quantum circuits. A key ingredient in the proof
of this correspondence is an interactive abstract machine (in the sense of~\cite{DanosRegnier:PNHil}) 
for the system, which is proved both to be a model of 
\QMLL\ cut-elimination and to give a computational meaning to proofs. This concrete approach, and the  close connection between 
circuits, logic and semantic models are an important difference with the current efforts in this direction.
The results we have obtained are encouraging.
%%%%%%%%%%%%%%%%%%%%%%%%%%%%%%%%%%%%%%%%%%%%%%%%%%
\section{Proof-Nets, GoI, and Superposition}
%%%%%%%%%%%%%%%%%%%%%%%%%%%%%%%%%%%%%%%%%%%%%%%%%%
The geometry of interaction of a proof, as initially conceived by Girard~\cite{Girard87a,Girard88}, is an operator on a 
Hilbert space $\ell^2$. As a matter of fact, however, the interpretation of linear logic proofs only makes use of a small 
fragment of the setting laid out in \cite{Girard88}. Fundamentally, the interpretation of a
proof (at least in the multiplicative fragment) is just a permutation on 
a finite set. Our aim here is to enrich linear logic in such a way that a larger portion 
of the GoI semantic universe is actually exploited, this way going towards 
a calculus with quantum features. The design of \QMLL\ has been guided by intuitions 
on what an hypothetical ``quantum GoI'' interpretation would look like.

In this section, we discuss some of the ideas which underlie GoI. The presentation is going to be very simplified; our purpose 
is to give an intuition rather than the formal details. For a more thorough presentation of GoI, 
we refer to \cite{Girard87a}, or to Girard's introductory notes \cite{GirardPA}.
We here focus our attention on multiplicative linear logic (\MLL\ in the following), and moreover only to 
\emph{cut-free} proofs. 

\newcommand{\tree}[1]{\mathsf{T}(#1)}
\paragraph{What is  a \MLL\ proof of  $\typeone$?}
Let us consider the (cut-free) proof-nets of \MLL, and make the assumption 
that all axioms are atomic. A cut-free proof of a formula 
$\typeone$ with $\natone$ occurrences of atoms 
will be interpreted as an $n\times n$ matrix (since the dimension is finite, 
we identify operators and matrices).
But how does a cut-free proof of a formula $\typeone$ look like?
Until we reach the axiom links, we have no freedom: in a proof-net, each formula 
is conclusion of a well-defined link, corresponding to the principal connective 
of the formula. In other words, a proof-net with conclusion $\typeone$ is necessarily the
disjoint union of two graphs:
\begin{varitemize}
\item  
  the formula tree $\tree{\typeone}$ of $\typeone$ (whose leaves are the occurrences of atoms);
\item 
  axiom links connecting pairs of dual atoms.
\end{varitemize}
All cut-free proofs of the same formula $\typeone$ have the same formula tree. What characterizes 
each of them is the linking among the occurrences of atoms. \emph{The atom links are hence enough to fully 
describe a cut-free \MLL\ proof of $\typeone$.}

As an example, let us consider the two cut-free proofs of the formula
$\typetwo=(\aa\b \parr \aa\b) \parr (\aa\otimes \aa)$. By 
indexing different occurrences of the same atom or co-atom, we
obtain $\typetwo=(\aa\b_1 \parr \aa\b_2) \parr (\aa_1\otimes \aa_2)$.
The two proofs (let us call them $\pfone$ and $\pftwo$)  
are in Figure \ref{fig:proofsB} below.
\begin{figure}[htbp]
\centering
\includegraphics{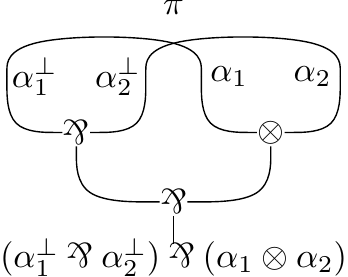}
\hspace{1cm}
\includegraphics{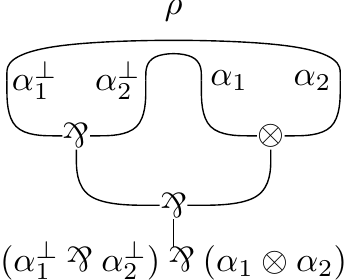}
\hspace{1cm}
\includegraphics{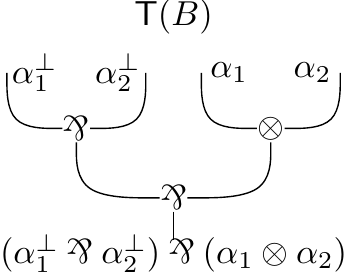}
\caption{Two proofs of $\typetwo$ and its syntax three $\tree{\typetwo}$}\label{fig:proofsB}
\end{figure}
Both  $\pfone$ and $\pftwo$ have the same formula tree $\tree{\typetwo}$.
In the case of $\pfone$, the axiom links are 
$\{\aa\b_1,\aa_1\}$ and $\{\aa\b_2,\aa_2\}$.
In the case of $\pftwo$, the axioms links are  
$\{\aa\b_1,\aa_2\}$ and $\{\aa\b_2,\aa_1\}$.

A convenient way to describe the links among $\natone$ (occurrences of) atoms is by means of a $n\times n$ matrix, which can be seen
as the adjacency matrix of the graph describing the axiom links, as the two graphs describing $\pfone$ and
$\pftwo$ in Figure~\ref{fig:graphs}. Since the axiom links describe the proof, this matrix itself is a faithful representation of the proof.
\begin{figure}[htbp]
\centering
\includegraphics{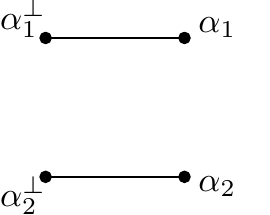}
\hspace{1cm}
\includegraphics{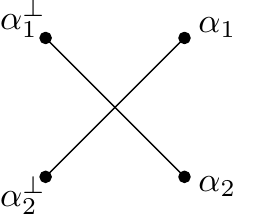}
\caption{Axiom graphs for $\pfone$ and $\pftwo$.}\label{fig:graphs}
\end{figure}     
Going back to our example, $\pfone$ and $\pftwo$ can be easily seen to be described by the following two matrices (once a 
suitable total order on atom and co-atom occurrences has been fixed):
$$
M=
\begin{pmatrix}
0 & 0 & 1 & 0 \\
0 & 0 & 0 & 1\\
1 & 0  & 0 & 0\\
0 & 1 & 0 &0
\end{pmatrix}
\qquad
N=
\begin{pmatrix}
0  & 0  & 0  & 1\\
0  & 0  & 1  & 0\\
0 &1  & 0  & 0\\
1  & 0 & 0 & 0
\end{pmatrix} 
$$
As a matter of fact, in this (simple, because the proofs are cut-free) case,
$M$ is actually the GoI interpretation of $\pfone$, while $N$ is 
the interpretation of $\pftwo$.

In general, the interpretation of a proof is
defined by induction, starting from the interpretation of the axioms.
For example, if $A$ is a formula with $n$ atoms, then
the axiom $\vdash A, A\b$ is associated to the  $2n \times 2n$  matrix 
$$
\begin{pmatrix}
 0 & I_n  \\
 I_n & 0
\end{pmatrix}
$$
where $I_n$ is the identity $n\times n$ matrix. We do not want to give more details here: the reader can find them 
clearly explained in~\cite{GirardPA}, Section 19.3.

\paragraph{Towards a quantum calculus.}
A matrix which interprets a \MLL\ proof is hermitian, because the links are undirected. It also has the 
following features: on the one hand it
has exactly one non-null element per row and column, and on the other hand
such an element is $1$. In other words, the matrix is a permutation matrix.
What if we relax these constraints?
Let us work again with an example. Let us consider the following $4\times 4$ hermitian complex matrix $M$:
$$
\begin{pmatrix}
0  & U \\
U^* &  0
\end{pmatrix}
$$
where U is the following $2\times2$ {\em unitary} matrix.
$$
\begin{pmatrix}
a_{11} & a_{12}  \\
a_{21} & a_{22}
\end{pmatrix}
$$
If we have two copies $\alpha\b_1,\alpha\b_2$ of $\alpha\b$ and two copies $\alpha_1,\alpha_2$ of $\alpha$, we can
see each non-zero coefficient $a_{ij}$ of $U$ as describing the existence of a link between the co-atom occurrence $\alpha\b_i$ 
and the atom occurrence $\alpha_j$;
the link is weighted by the coefficient $a_{ij}$.
Hence we read the matrix $M$ above as the interpretation of the following ``weighted'' set
of axiom links:
\begin{center}
\includegraphics{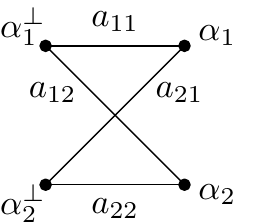}
\end{center}

We think of this set of  links as being in  ``quantum superposition'': we have the link 
$\{\aa\b_1,\aa_1\}$ with amplitude $a_{11}$, and the link $\{\aa\b_1, \aa_2 \}$ with amplitude $a_{12}$.
%(we will make this clear in the following).

Such a ``generalized axiom'' can be  described in a compact way by providing a pair:
an atomic axiom link %$($\vdash \aa,\aa\b$)$
and the unitary matrix $U$. The graph above, in other words, becomes the following proof
\begin{center}
\includegraphics{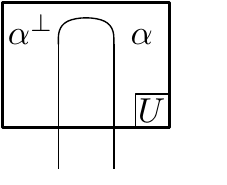}
\end{center}
namely something like a ``box'', labeled with the unitary matrix $U$ and containing an axiom link.
This is actually the idea beyond the $\qrule{n}$-rule, which characterizes our calculus $\QMLL$ with respect to
ordinary, classical, $\MLL$ (see next section). More generally, an atomic axiom link and a unitary matrix 
$V=(a_{ij})$ on  $\CC^{2^n}$, describe a generalized axiom link, which consists in:
\begin{varitemize}
\item 
  $2^n$ occurrences of $\bt{\alpha}$ and $2^n$ occurrences of $\alpha$;
\item
  $2^{2n}$ links; the complex number $a_{ij}$ from $V$ describes the presence of a link 
  with amplitude $a_{ij}$ from the occurrence $\bt{\alpha}_i$ to the occurrence
  $\alpha_j$.
\end{varitemize}
Logically, the intuition is that the rule $\qrule{n}$ produces $2^{2n}$ copies of the axiom link, which are in  
``quantum superposition''. Actually, rule $\qrule{n}$ acts not only on axiom links, but 
on arbitrary proofs.

The operator-theoretic GoI interpretation of \QMLL, as well as a formal development of a system of proof-nets for it, is 
not the object of this preliminary report; we postpone them to a follow-up paper. We briefly explained it here
because these are the main inspiring ideas behind \QMLL.
%%%%%%%%%%%%%%%%%%%%%%%%%%%%%%%%
\section{The Syntax of \QMLL}
%%%%%%%%%%%%%%%%%%%%%%%%%%%%%%%%
Formulas of \QMLL\ are generated by the following grammar:
$$
\typeone ::=\atomone\mmid\bt{\atomone}\mmid
\typeone\parr\typeone \mmid\typeone\tens\typeone\mmid \qcn{A}\mmid\bqcn{A}
$$
In other words, \QMLL's formulas are obtained by enriching the language
of \MLL\ with two unary modal connectives, namely $\qcc$ and $\bqcc$, which are
dual of each other. Linear negation can  then be defined
 in the usual way, by setting $\bt{(\qcn{A})}=\bqcn{\bt A}$ and $\bt{(\bqcn{A})}=\qcn{\bt A}$.
 As an example, $\bt{(\qcn{(\typeone\parr\typetwo)})}=
\bqcn{(\bt\typeone\otimes\bt\typetwo)}$.
$\qcnp{\typeone}{\natone}$ is syntactic sugar for
$$\underbrace{\boxdot(\boxdot(\ldots\boxdot(}_{\mbox{$\natone$ times}}\typeone)\ldots)).$$
Similarly for $\bqcnp{\typeone}{\natone}$.

A \emph{modal formula} is a formula in the form $\qcn{\typeone}$ or $\bqcn{\typeone}$.
We reserve the metavariables $\qtypeone,\qtypetwo$ to indicate modal formulas, and 
the metavariables $\mtypeone,\mtypetwo$ to indicate formulas whose most-external connective is not a modality ($\qcc$ or $\bqcc$).

\QMLL\ will be given as a sequent calculus.
\emph{Sequents} have the form $\jd{\ctxone}$, where 
 $\ctxone,\ctxtwo$ are finite multisets of formulas. 
$\unop{\natone}$ is the class of unitary operators from
 $\CC^{2^\natone}$ to itself. We denote by $\idmat{\natone}$  the identity operator
on $\unop{\natone}$. The following are the rules of \QMLL.
$$
\begin{array}{ccccc}
  \infer[\axiom]
  {\jd{\bt{\typeone}, \typeone }}
  {}
  &
  &
  \infer[\cut]
  {\jd{\ctxone,\ctxtwo}}
  {\jd{\ctxone,\typeone} & \jd{\Delta,\bt{\typeone}}}
  &
  &
  \infer[\parr]
  {\jd{\ctxone,\typeone\parr \typetwo}}
  {\jd{\ctxone,\typeone,\typetwo }}
\end{array}
$$
\vspace{1mm}
$$
\begin{array}{ccccc}
  \infer[\tens]
  {\jd{\ctxone,\ctxtwo,\typeone\tens\typetwo}}
  {\jd{\ctxone,\typeone} & \jd{\ctxtwo,\typetwo}}
  &
  &
  \infer[\qqrule{n}]
  {\jd{\qqb{\qtypeone}{n},\qq{\qtypetwo}{n}}}
  {\jd{\qtypeone,\qtypetwo} & \unone\in\unop{n}}
  &
  &
  \infer[\mqrule{n}]
  {\jd{\qqb{\mtypeone}{n},\qq{\mtypetwo}{n}}}
  {\jd{\mtypeone,\mtypetwo} & \unone\in\unop{n}}
\end{array}
$$
The rules $\qqrule{n}$ and $\mqrule{n}$
are said to be \emph{quantum rules}.
Observe how the two quantum  rules act exactly the same way on
their premise, adding $n$ instances of the modalities $\qcc$ and $\bqcc$ to each of the two
formulas in it. The only difference is in the nature of those formulas, which are
required to be modal  formulas in $\qqrule{n}$ and not modal formulas in $\mqrule{n}$. In the following, 
$\qrule{n}$ stands for either $\qqrule{n}$ or $\mqrule{n}$. 

\newcommand{\MELL}{\textsf{MELL}}
A \emph{proof} $\pfone$ is, as usual, a tree built according to the rules above. Occurrences
of quantum rules can be seen as \emph{boxes}, similarly to what happens with exponential
modalities, e.g., in \MELL.

The \emph{principal formulas} of an instance of a rule are the occurrences of formulas 
which are introduced (or cut) by the rule. Observe how any instance of $\qrule{n}$ has two principal formulas 
(and no other formula).
%%%%%%%%%%%%%%%%%%%%%%%%%%%%%
\section{Cut Elimination}\label{sect:cutelim}
%%%%%%%%%%%%%%%%%%%%%%%%%%%%%
In this section, \QMLL\ will be proved to enjoy cut-elimination. This will be carried out by giving 
an effective binary relation on proofs which allows to remove all instances of rule $\cut$
from proofs.

Formally, the relation $\becomes$ on the space of \QMLL\ proofs is defined by some \emph{reduction rules}, which
can be applied in any context:
\begin{varitemize}
\item
  \textbf{Axiom Reduction.}  
  Every time one of the premises of a cut-rule is an axiom, the cut is eliminated in the usual way,
  by means of the following reduction:
  $$ 
  \infer[\cut]{\vdash \Gamma, A}{{\Gamma, A }  &  \infer[\axiom]{\vdash A\b,A}{}} 
  \becomes
  {\vdash \Gamma, A }{}  
  $$
\item
  \textbf{Multiplicative Principal Reduction.}
  The dual multiplicative connectives $\tens$ and $\parr$ annihilate each other, as usual:
  $$
  \infer[\cut]
  {\jd{\ctxone,\ctxtwo,\ctxthree}}
  {\infer[\tens]
    {\jd{\ctxone,\ctxtwo,\typeone\tens\typetwo}}
    {\jd{\ctxone,\typeone}  & \jd{\ctxtwo,\typetwo}}
   & 
   \infer[\parr]
    {\jd{\ctxthree,\bt{\typeone}\parr\bt{\typetwo}}}
    {\jd{\ctxthree,\bt{\typeone},\bt{\typetwo}}}}
  \becomes
  \infer[\cut]
  {\jd{\ctxone,\ctxtwo,\ctxthree}}
  {\jd{\ctxone,\typeone}
   &
   \infer[\cut]
   {\jd{\ctxtwo,\ctxthree,\bt{\typeone}}}
   {\jd{\ctxtwo,\typetwo} & \jd{\ctxthree,\bt{\typeone},\bt{\typetwo}}}
  }
  $$
\item
  \textbf{Quantum Principal Reduction.}
  This reduction can be performed  when
  both cut-formulas are introduced by the rule
  $\qrule{m}$  (the arity $m$ being the same in both sides):
  $$
  \infer[\cut]
  {\jd{ \qqb{\typeone}{m},  \qq{\typethree}{m}   }}
  {
    \infer[\qrule{m}]
    { \jd{ \qqb{\typeone}{m}, \qq{\typetwo}{m}  }}
    {\jd{\typeone,\typetwo} & \unone \in\unop{m}}
    &
    \infer[\qrule{m}]
    {  \jd{ \qqb{\bt{\typetwo}}{m},\qq{\typethree}{m}}  }
    {\jd{\bt{\typetwo},\typethree} & \untwo\in\unop{m}      }
  }
  \becomes
  \infer[\qrule{m}]
  {   \jd{\qqb{\typeone}{m},\qq{\typethree}{m}}}
  {
    \infer[\cut]
    {\jd{\typeone,\typethree}}
    {
      \jd{\typeone,\typetwo}
      &
      \jd{\bt{\typetwo},\typethree}
    }
    &
    \unone\cdot\untwo\in\unop{m}}
  $$
\item
\textbf{Quantum $\eta$-Expansion.} Axioms introducing modal formulas
can be $\eta$-expanded as follows:
  $$
  \infer[\axiom]
  {\qqb{\bt{\typeone}}{n},\qq{\typeone}{n}}
  {}
  \becomes
  \infer[\qrule{n}]
  {\jd{\qqb{\bt{\typeone}}{n},\qq{\typeone}{n}}}
  {\infer[\axiom]{\jd{\bt{\typeone},\typeone}}{} & I_n\in\unop{n}}
  $$
\item
\textbf{Quantum Contraction.} Two successive applications of a quantum rule can be contracted:
$$
\infer[\qrule{n}]
{ \jd{\qqb{A}{k+n},\qq{B}{k+n}}}{\infer[\qrule{k}]{\qqb{A}{k},\qq{B}{k}}
{\jd{A,B}  & \unone\in \unop{k}} &  \untwo\in \unop{n}} 
\becomes
\infer[\qrule{k+n}]
{\jd{\qqb{A}{k+n},\qq{B}{k+n}}}
{\jd{A,B} & \unone\tens\untwo\in\unop{k+n}}
$$
\item
\textbf{Commuting Reduction.}
  These three reduction rules allow us to lift up a cut whose principal formula is not introduced immediately over it:
  $$
  \infer[\cut]
  {\jd{\ctxone,\ctxtwo,\typetwo\parr\typethree}}
  {\jd{\ctxone,\typeone}
   & 
   \infer[\parr]
    {\jd{\ctxtwo,\bt{\typeone},\typetwo\parr\typethree}}
    {\jd{\ctxtwo,\bt{\typeone},\typetwo,\typethree}}}
  \becomes
  \infer[\parr]
    {\jd{\ctxone,\ctxtwo,\typetwo\parr\typethree}}
    {\infer[\cut]
      {\jd{\ctxone,\ctxtwo,\typetwo,\typethree}}
      {\jd{\ctxone,\typeone} &
       \jd{\ctxtwo,\bt{\typeone},\typetwo,\typethree}}}
  $$
  $$
  \infer[\cut]
  {\jd{\ctxone,\ctxtwo,\ctxthree,\typetwo\tens\typethree}}
  {\jd{\ctxone,\typeone}
   & 
   \infer[\tens]
    {\jd{\ctxtwo,\ctxthree,\bt{\typeone},\typetwo\tens\typethree}}
    {\jd{\ctxtwo,\bt{\typeone},\typetwo} & \jd{\ctxthree,\typethree}}}
  \becomes
  \infer[\tens]
    {\jd{\ctxone,\ctxtwo,\ctxthree,\typetwo\tens\typethree}}
    {\infer[\cut]
      {\jd{\ctxone,\ctxtwo,\typetwo}}
      {\jd{\ctxone,\typeone} &
       \jd{\ctxtwo,\bt{\typeone},\typetwo}}
     &
     \jd{\ctxthree,\typethree}}
  $$
 $$
  \infer[\cut]
  {\jd{\ctxone,\ctxtwo,\ctxthree,\typetwo\tens\typethree}}
  {\jd{\ctxone,\typeone}
   & 
   \infer[\tens]
    {\jd{\ctxtwo,\ctxthree,\bt{\typeone},\typetwo\tens\typethree}}
    {\jd{\ctxtwo,\typetwo} & \jd{\ctxthree,\bt{\typeone},\typethree}}}
  \becomes
  \infer[\tens]
    {\jd{\ctxone,\ctxtwo,\ctxthree,\typetwo\tens\typethree}}
    {\infer[\cut]
      {\jd{\ctxone,\ctxthree,\typethree}}
      {\jd{\ctxone,\typeone} &
       \jd{\ctxthree,\bt{\typeone},\typethree}}
     &
     \jd{\ctxtwo,\typetwo}}
  $$
\end{varitemize}
\newcommand{\becomess}{\becomes^*}
Let $\becomess$ be the reflexive and transitive closure of $\becomes$. A proof $\pfone$ is \emph{normal}
if there is not any $\pftwo$ such that $\pfone\becomes\pftwo$.
\begin{proposition}
Every normal proof is cut-free.
\end{proposition}
\begin{proof}
We can prove that any proof $\pfone$ containing a cut is not normal by induction on the structure
of $\pfone$. The only interesting case is the one when the last rule of $\pfone$ is a cut and
the two immediate sub-proofs are cut-free. All the other cases can be easily handled by way of the
induction hypothesis. In other words, $\pfone$ is
$$
\infer[\cut]
  {\jd{\ctxone,\ctxtwo}}
  {\jdp{\pftwo}{\ctxone,\typeone} & \jdp{\pfthree}{\ctxtwo,\bt{\typeone}}}
$$
where $\pftwo$ an $\pfthree$ are cut-free. Now:
\begin{varitemize}
\item
  If either $\pftwo$ or $\pfthree$ is introduced by an axiom, $\pfone$ is not normal, since an
  axiom reduction can be applied to it;
\item
  If the last rule in $\pftwo$ is multiplicative, then:
  \begin{varitemize}
  \item
    If the principal formula of that rule is not $\typeone$, then a commuting
    reduction can be applied to $\pfone$.
  \item
    If the principal formula of that rule is precisely $\typeone$, then consider the last rule of
    $\pfthree$. If it is multiplicative and has principal formula $\bt{\typeone}$, then $\pfone$
    is not normal, because a multiplicative principal reduction can be applied to it. If it 
    is multiplicative and has a principal formula in $\ctxtwo$, then
    again $\pfone$ is not normal, because a commuting reduction can be applied to it.
  \end{varitemize}
\item
  If the last rule in $\pfthree$ is multiplicative, then we can proceed as in the previous case;
\item
  We can then assume that $\pfone$ has the following form:
  $$
  \infer[\cut]
  {\jd{\qqb{\typeone}{m},\qq{\typefour}{n}}}
  {\infer[\qrule{m}]
    {\jd{\qqb{\typeone}{m},\qq{\typetwo}{m}}}
    {\jdp{\pffour}{\typeone,\typetwo}}
    & 
   \infer[\qrule{n}]
    {\jd{\qqb{\typethree}{n},\qq{\typefour}{n}}}
    {\jdp{\pffive}{\typethree,\typefour}}}
  $$
  If $n=m$, then we perform a quantum principal reduction. Otherwise, let assume $m>n$. In this case, 
  both $\typethree$ and $\typefour$ must be modal formulas, because $\qqb{\typethree}{n}=\bt{(\qq{\typetwo}{m})}$ 
  and, as a consequence, the last rule in $\pffive$ must
  be itself a quantum rule, or an axiom. Hence we can perform either a quantum contraction, or (in case of axiom) a quantum expansion.
\end{varitemize}
This concludes the proof.
\end{proof}

\begin{proposition}\label{CSN}
The binary relation $\becomes$ is confluent and strongly normalizing.
\end{proposition}
\begin{proof}
Actually, $\becomes$ is \emph{strongly} confluent, as can be proved by analyzing the different cases.
The fact $\becomes$ is strongly normalizing can be proved by attributing a weight to any rule and by
showing that the total weight of a proof $\pfone$ (i.e. the sum of the weight of all rule instances
in $\pfone$) strictly decreases along $\becomes$.
\end{proof}

%%%%%%%%%%%%%%%%%%%%%%%%%%%%%%%%%%%%%%%%%%%
\section{Encoding of Quantum Circuits}\label{sec:QCencoding}
%%%%%%%%%%%%%%%%%%%%%%%%%%%%%%%%%%%%%%%%%%%
\emph{Quantum circuits} are an efficient model of quantum computation. A quantum circuit is an acyclic network 
of quantum gates connected by wires, where each gate represents an operation on the qubits on which the gate acts. 
In this paper,  we  are interested in {\em unitary quantum circuits}, i.e. circuits in which all the gates 
correspond to unitary operations.  It is a standard result that a general quantum circuit can be simulated 
by a unitary quantum circuit (its {\em unitary purification}) plus some ancillary qubits, to be measured 
or ignored at the end of the computation.

In this section, we give some  intuitions about the quantum modality, by illustrating the fact that 
all unitary quantum circuits are captured by \QMLL\ proofs.
Let us  consider proofs in \QMLL\  which do not make use of multiplicative connectives.
We can see a proof of conclusion $\jd{\bqcnp{\bt{\atomone}}{n},\qcnp{\atomone}{n}}  $  as a circuit on $n$ qubits.
The number of occurrences of modalities in which $\alpha^\bot$ (resp. $\alpha$) is nested in the conclusion, 
corresponds to the number of qubits on which the circuit acts.
More precisely, at depth 1 we have the first qubit, at depth 2 the second, and so on.
We can then act on the qubits from $j+1$ to $j+k$, by applying  
a quantum rule $\qrule{k}$ to $\jd{\bqcnp{\bt{\atomone}}{j},\qcnp{\atomone}{j}}  $.
This is best illustrated by some examples. Let us start with a simple circuit on 3 qubits 
(and no operation on them), and its encoding:

{\footnotesize
 $$
 \begin{array}{ccc}
\begin{minipage}{3cm}
 \Qcircuit @C=1em @R=.7em { 
 & \qw & \qw \\ 
 & \qw & \qw \\ 
  & \qw & \qw } 
\end{minipage} 
&
  \qquad
& 
\begin{minipage}{3cm}
 \infer
   {\vdash{\bqcnp{\bt{\alpha}}{3}, \qcnp{\alpha}{3}}}
   {\vdash {\bt \alpha, \alpha}    &  I\otimes I \otimes I\in\unop{3}}
 \end{minipage}
 \end{array}
 $$
 }

\noindent The application of the  Hadamard gate on the second qubit can be represented as follows:

 {\footnotesize
 $$
\begin{array}{ccc}
\begin{minipage}{3cm}
 \Qcircuit @C=1em @R=.7em { 
 & \qw &  \qw \\
 & \gate{H} &\qw \\ 
 & \qw &  \qw }
 
\end{minipage}
&\qquad &
\begin{minipage}{3cm}
\infer
   {\vdash{\bqcnp{\bt{\alpha}}{3}, \qcnp{\alpha}{3}}}
{   \infer
     {\vdash{\bqcnp{\bt{\alpha}}{2}, \qcnp{\alpha}{2}}}
{\infer  {\vdash{\bqcnp{\bt{\alpha}}{1}, \qcnp{\alpha}{1}}}
  {\vdash {\bt \alpha, \alpha}    &  I\in\unop{1}}
 &  H\in\unop{1} }
 & I\in\unop{1}}
 
\end{minipage}
\end{array}
$$
}
 
\noindent Let us now represent a circuit which  applies Hadamard to the first qubit, and the CNOT gate to the second and third qubits:

 {\footnotesize
   $$
   \begin{array}{ccc}
    \begin{minipage}{3cm}
    \Qcircuit @C=1em @R=.7em { 
    & \gate{H} &\qw \\ 
     & \ctrl{1} & \qw \\ 
     & \targ & \qw }
  \end{minipage}
 &\qquad &
  \begin{minipage}{3cm}
  \infer
      {\vdash{\bqcnp{\bt{\alpha}}{3}, \qcnp{\alpha}{3}}}
    {\infer
      {\vdash{\bqcnp{\bt{\alpha}}{1}, \qcnp{\alpha}{1}}}
      {\vdash {\bt \alpha, \alpha}    &  H \in\unop{1}}  
      & CNOT \in\unop{2}
      }
       \end{minipage}
         \end{array}
    $$
 }

\noindent Applying a gate after the other to the same qubit(s) corresponds to composing the unitary operators,
which is naturally performed by the cut-rule.
As an example, the quantum circuit $\qcone$
graphically represented in Figure \ref{fig:exqc} can be encoded as follows:
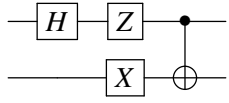
\begin{figure}
$$
\Qcircuit @C=1em @R=.7em { 
& \gate{H} & \gate{Z} & \ctrl{1} & \qw \\ 
& \qw & \gate{X} & \targ & \qw }
$$ 
\caption{An Example Quantum Circuit}\label{fig:exqc}
\end{figure}

{\footnotesize
$$
    \infer
    {\jd{\bqcnp{\bt{\atomone}}{2},\qcnp{\atomone}{2}}}
    {
      \infer
      {\jd{\bqcnp{\bt{\atomone}}{2},\qcnp{\atomone}{2}}}
      {
        \infer[\qrule{1}]
        {\jd{\bqcnp{\bt{\atomone}}{2},\qcnp{\atomone}{2}}}
        {
          \infer[\qrule{1}]
          {\jd{\bqcnp{\bt{\atomone}}{1},\qcnp{\atomone}{1}}}
          {
            \jd{\bt{\atomone},\atomone}
            &
            H\in\unop{1}
          }
          &
          \idmat{}\in\unop{1}
        }
        &
        \infer[\qrule{1}]
        {\jd{\bqcnp{\bt{\atomone}}{2},\qcnp{\atomone}{2}}}
        {
          \infer[\qrule{1}]
          {\jd{\bqcnp{\bt{\atomone}}{1},\qcnp{\atomone}{1}}}
          {
            \jd{\bt{\atomone},\atomone}
            &
            Z\in\unop{1}
          }
          &
          X\in\unop{1}
        }
      }
      &
      \infer[\qrule{2}]
      {\jd{\bqcnp{\bt{\atomone}}{2},\qcnp{\atomone}{2}}}
      {
        \jd{\bt{\atomone},\atomone}
        &
        {CNOT}\in\unop{2}
      }
    }
 % }
$$}

\noindent With this, it is easy to see that we can  faithfully capture
any unitary quantum circuit $\qcone$ acting on $m$ qubits by a \QMLL\ proof $\pfone_\qcone$ with
conclusion $\jd{\qqb{\bt{\atomone}}{m}, \qq{\atomone}{m}}$.
Conversely, given a generic \QMLL\ proof, we can retrieve a set of quantum circuits. The next section is devoted
to formalizing and proving this claim. This is done by introducing an abstract machine, which given a 
quantum register and a \QMLL\ derivation, applies to the register the operations coded in the proof.

%%%%%%%%%%%%%%%%%%%%%%%%%%%%%%%%%%%%%%%%%%%%%%%%%%
\section{The Quantum Interaction Abstract Machine}
%%%%%%%%%%%%%%%%%%%%%%%%%%%%%%%%%%%%%%%%%%%%%%%%%%
It is natural to wonder if there is any computational interpretation of the cut elimination procedure we introduced in 
Section~\ref{sect:cutelim}. The classical, multiplicative, portion of \QMLL\ behaves as usual:
$\parr$ and $\tens$ are dual connectives which interact in a purely classical fashion by
annihilating each other. This corresponds to (linear) beta reduction in the lambda calculus.
But what about the new modalities $\qcc$ and $\bqcc$? The reduction rules involving them, namely
quantum principal reduction, quantum $\eta$-expansion and quantum contraction 
correspond to various ways of creating and aggregating
unitary transformations (i.e. quantum gates). However, what we would like to have in order to talk of quantum computation, is the ability to act on a quantum register.

In this section, we describe an interpretation of \QMLL\ proofs in terms of an automata-based
view of the Geometry of Interaction due to Danos and Regnier~\cite{DanosRegnier:PNHil}, called the Quantum Interaction
Abstract Machine (\QIAM\ in the following). This, in particular, will constitute a concrete
computational interpretation of \QMLL\ cut elimination, being a model of it.\\

In order to give the machine, we need to introduce some technical definitions.
A \emph{context} is simply a formula with an ``hole'' $\ec$ in it:
$$
\conone ::=\ec\mmid\conone\parr\typeone\mmid\typeone\parr\conone\mmid\conone\tens\typeone\mmid\typeone\tens\conone\mmid
\qcn{\conone}\mmid\bqcn{\conone}.
$$
The formula obtained by substituting $\typeone$ for $\ec$ in a context $\conone$ is indicated as $\csb{\conone}{\typeone}$.
If $\typeone=\csb{\conone}{\atomone}$ (respectively, if $\typeone=\csb{\conone}{\bt{\atomone}}$), then 
$\conone$ is said to be a \emph{positive} (respectively, a \emph{negative}) \emph{context for} $\typeone$.
If $\conone$ is either positive or negative for $\typeone$, then we simply say that $\conone$ is a \emph{context for}
$\typeone$. To emphasize that a context $\conone$ is positive (negative, respectively) for a formula $\typeone$,
we indicate it with the metavariable  $\pconone$  (respectively, $\nconone$). Given
a context $\conone$, its dual $\bt{\conone}$ can be easily defined, e.g. $\bt{(\typeone\tens\conone)}=\bt{\typeone}\parr\bt{\conone}$.

The {\em nesting depth} of $\conone$, denoted $\depth{\conone}$, is the number of occurrences of modal operators in which
$\ec$ is embedded. A \emph{stack} is an element of $\{\qcc,\bqcc\}^*$, i.e., a finite sequence of elements of $\{\qcc,\bqcc\}$, each
seen as an atomic symbol. 

The {\em quantum interactive abstract machine} $\qiam{\pfone}$ associated to any proof $\pfone$ consists in:
\begin{varenumerate}
\item
  The set of \emph{states} $\sts{\pfone}$, which contains all the
  quadruples in the form $(\typeone,\conone,\stckone,\regone)$,
  where $\typeone$ is an occurrence of a formula in $\pfone$, $\conone$ is 
  a context for $\typeone$, $\stckone$ is a stack, 
  and $\regone$ is a \emph{quantum register} of $\depth{\conone}+\length{\stckone}$ qubits.
\item 
  A \emph{transition relation} $\trs{\pfone}\subseteq\sts{\pfone}\times\sts{\pfone}$.
\end{varenumerate}
The transition relation $\trs{\pfone}\subseteq\sts{\pfone}\times\sts{\pfone}$ is defined
by way of a set of rules which we will introduce shortly. Before doing that, let us remark
that:
\begin{varitemize}
\item
  We can see the  transition rules as providing  instructions for a token to travel around the proof. With the transition 
  $(\typeone,\conone,\stckone,\regone)\trs{\pfone}(\typetwo,\contwo,\stcktwo,\regtwo)$, the token  moves from the (occurrence of)
  formula $\typeone$ to the (occurrence of) formula $\typetwo$. In general, $\typeone$ and $\typetwo$ appear in sequents which 
  are one on top of the other  (i.e., premise and conclusion of the same rule); the only  exception is when $\typeone$ is 
  the principal formula of a cut or an axiom: in such a case $\typetwo$ will be the principal formula $\bt{\typeone}$. 
  In the case when $\typeone$ and $\typetwo$ appear in sequents which are one on top of
  the other,  if $\conone$ is positive (negative, respectively) for $\typeone$, then the token goes \emph{down} 
  (\emph{up}, respectively).  When on axioms and on principal formulas of the cut rule, 
  the token  \emph{inverts its direction}. 
\item 
  The size $\depth{\conone}+\length{\stckone}$ of the quantum register is constant. 
  We operate on the the quantum register $\regone$ only when exiting from a quantum box. At that moment,
  the unitary transformation associated to the box (or its inverse) is applied to $\regone$.
\item
  The r\^ole of the context $\conone$ is similar to the one of the multiplicative stack in ordinary \IAM, while
  the r\^ole of $\stckone$ consists in keeping track of which of the two ports of boxes have been traversed
  to reach the current position. In other words, the length of the stack is exactly the
  ``box-nesting depth'' of $\typeone$ in $\pfone$.
\end{varitemize}
The rules defining $\trs{\pfone}$ are indeed independent on the specific structure of $\pfone$
and, instead, only depend on the six proof rules of \QMLL. They are in Figure \ref{Transitions} and are given 
in an informal but hopefully intuitive way.
\begin{figure}
{\footnotesize
\begin{center}
Rule $\axiom$\\
\begin{tabular}{cc}
\begin{minipage}[!t]{7cm}
\begin{eqnarray*}
(\typeone,\nconone,\stckone,\regone)&\trs{\pfone}&(\bt{\typeone},\bt{\nconone},\stckone,\regone)\\
(\bt{\typeone},\nconone,\stckone,\regone)&\trs{\pfone}&(\typeone,\bt{\nconone},\stckone,\regone)
\end{eqnarray*}
\end{minipage}
&
\begin{minipage}{5cm}
$$
\infer[]
{\jd{\bt{\typeone}, \typeone }}
{}
$$
\end{minipage}
\end{tabular}
\end{center}
\begin{center}
Rule $\cut$\\
\begin{tabular}{cc}
\begin{minipage}{7cm}
\begin{eqnarray*}
(\typeone,\pconone,\stckone,\regone)&\trs{\pfone}&(\bt{\typeone},\bt{\pconone},\stckone,\regone)\\
(\bt{\typeone},\pconone,\stckone,\regone)&\trs{\pfone}&(\typeone,\bt{\pconone},\stckone,\regone)\\
(\ctxone_1,\nconone,\stckone,\regone)&\trs{\pfone}&(\ctxone_2,\nconone,\stckone,\regone)\\
(\ctxtwo_1,\nconone,\stckone,\regone)&\trs{\pfone}&(\ctxtwo_2,\nconone,\stckone,\regone)\\
(\ctxone_2,\pconone,\stckone,\regone)&\trs{\pfone}&(\ctxone_1,\pconone,\stckone,\regone)\\
(\ctxtwo_2,\pconone,\stckone,\regone)&\trs{\pfone}&(\ctxtwo_1,\pconone,\stckone,\regone)
\end{eqnarray*}
\end{minipage}
&
\begin{minipage}{5cm}
$$
\infer[]
{\jd{\ctxone_1,\ctxtwo_1}}
{\jd{\ctxone_2,\typeone} & \jd{\ctxtwo_2,\bt{\typeone}}}
$$
\end{minipage}
\end{tabular}
\end{center}
\begin{center}
Rule $\parr$\\
\begin{tabular}{cc}
\begin{minipage}{7cm}
\begin{eqnarray*}
(\typeone\parr\typetwo,\nconone\parr\typetwo,\stckone,\regone)&\trs{\pfone}&(\typeone,\nconone,\stckone,\regone)\\
(\typeone\parr\typetwo,\typeone\parr\nconone,\stckone,\regone)&\trs{\pfone}&(\typetwo,\nconone,\stckone,\regone)\\
(\typeone,\pconone,\stckone,\regone)&\trs{\pfone}&(\typeone\parr\typetwo,\pconone\parr\typetwo,\stckone,\regone)\\
(\typetwo,\pconone,\stckone,\regone)&\trs{\pfone}&(\typeone\parr\typetwo,\typeone\parr\pconone,\stckone,\regone)\\
(\ctxone_1,\nconone,\stckone,\regone)&\trs{\pfone}&(\ctxone_2,\nconone,\stckone,\regone)\\
(\ctxone_2,\pconone,\stckone,\regone)&\trs{\pfone}&(\ctxone_1,\pconone,\stckone,\regone)
\end{eqnarray*}
\end{minipage}
&
\begin{minipage}{5cm}
$$
  \infer[]
  {\jd{\ctxone_1,\typeone\parr \typetwo}}
  {\jd{\ctxone_2,\typeone,\typetwo }}
$$
\end{minipage}
\end{tabular}
\end{center}
\begin{center}
Rule $\tens$\\
\begin{tabular}{cc}
\begin{minipage}{7cm}
\begin{eqnarray*}
(\typeone\tens\typetwo,\nconone\tens\typetwo,\stckone,\regone)&\trs{\pfone}&(\typeone,\nconone,\stckone,\regone)\\
(\typeone\tens\typetwo,\typeone\tens\nconone,\stckone,\regone)&\trs{\pfone}&(\typetwo,\nconone,\stckone,\regone)\\
(\typeone,\pconone,\stckone,\regone)&\trs{\pfone}&(\typeone\tens\typetwo,\pconone\tens\typetwo,\stckone,\regone)\\
(\typetwo,\pconone,\stckone,\regone)&\trs{\pfone}&(\typeone\tens\typetwo,\typeone\tens\pconone,\stckone,\regone)\\
(\ctxone_1,\nconone,\stckone,\regone)&\trs{\pfone}&(\ctxone_2,\nconone,\stckone,\regone)\\
(\ctxtwo_1,\nconone,\stckone,\regone)&\trs{\pfone}&(\ctxtwo_2,\nconone,\stckone,\regone)\\
(\ctxone_2,\pconone,\stckone,\regone)&\trs{\pfone}&(\ctxone_1,\pconone,\stckone,\regone)\\
(\ctxtwo_2,\pconone,\stckone,\regone)&\trs{\pfone}&(\ctxtwo_1,\pconone,\stckone,\regone)
\end{eqnarray*}
\end{minipage}
&
\begin{minipage}{5cm}
$$
  \infer[]
  {\jd{\ctxone_1,\ctxtwo_1,\typeone\tens\typetwo}}
  {\jd{\ctxone_2,\typeone} & \jd{\ctxtwo_2,\typetwo}}
$$
\end{minipage}
\end{tabular}
\end{center}
\begin{center}
Rule $\qrule{n}$\\
\begin{tabular}{cc}
\begin{minipage}{9cm}
\begin{eqnarray*}
(\qqb{\typeone}{n},\qqb{\nconone}{n},\stckone,\regone)&\trs{\pfone}&(\typeone,\nconone,\stckone\cdot\bqcc^n,\regone)\\
(\qq{\typetwo}{n},\qq{\nconone}{n},\stckone,\regone)&\trs{\pfone}&(\typeone,\nconone,\stckone\cdot\qcc^n,\regone)\\
(\typeone,\pconone,\stckone\cdot\bqcc^n,\regone)&\trs{\pfone}&(\qqb{\typeone}{n},\qqb{\pconone}{n},\stckone,\regone)\\
(\typeone,\pconone,\stckone\cdot\qcc^n,\regone)&\trs{\pfone}&(\qqb{\typeone}{n},\qqb{\pconone}{n},\stckone,
(\idmat{\depth{P}} \otimes\unone^* \otimes \idmat{\length{\stckone}} )
(\regone))\\
(\typetwo,\pconone,\stckone\cdot\qcc^n,\regone)&\trs{\pfone}&(\qq{\typetwo}{n},\qq{\pconone}{n},\stckone,\regone)\\
(\typetwo,\pconone,\stckone\cdot\bqcc^n,\regone)&\trs{\pfone}&(\qq{\typetwo}{n},\qq{\pconone}{n},\stckone,
(\idmat{\depth{P}} \otimes\unone \otimes \idmat{\length{\stckone}}  )(\regone))\\
\end{eqnarray*}
\end{minipage}
&
\begin{minipage}{5cm}
$$
  \infer[\qrule{n}]
  {\jd{\qqb{\typeone}{n},\qq{\typetwo}{n}}}
  {\jd{\typeone,\typetwo} & \unone\in\unop{n}}
$$
\end{minipage}
\end{tabular}
\end{center}}
\caption{Defining Rules for $\trs{\pfone}$}\label{Transitions}
\end{figure}
As it can be easily seen, the relation $\trs{\pfone}$ is a partial function: for every
state $\stone$ there is \emph{at most} one state $\sttwo$ such that $\stone\trs{\pfone}\sttwo$.
Moreover, it is an injection: no two states $\stone,\sttwo$ can lead
to the same $\stthree$ via $\trs{\pfone}$. 

Now, let us turn our attention to the way the quantum register $\regone$ is manipulated during
computation. As previously observed, the only way to alter the value of the quantum register consists in exiting from
a box. Moreover, the way any state $(\typeone,\nconone,\stckone,\regone)$ evolves does \emph{not}
depend on $\regone$: the operations applied to the underlying quantum register along
$\trs{\pfone}$ only depend on the first three components of the state. 
The current value of the quantum register has no effect on the value of the first three components 
after any transition. This is captured by the following:
\begin{lemma}[Uniformity]\label{lemma:cc}
For every proof $\pfone$ and for every $\typeone,\conone,\stckone$, there are
$\typetwo,\contwo,\stcktwo$ and a unitary operator
$\unone$ on $\CC^{2^{\depth{\conone}+\length{\stcktwo}}}$ such
that for every $\regone$, if
$(\typeone,\conone,\stckone,\regone)\trs{\pfone}(\typethree,\conthree,\stckthree,\regtwo)$
then $\typethree=\typetwo$, $\conthree=\contwo$, $\stckthree=\stcktwo$ and $\regtwo=\unone(\regone)$.
\end{lemma}

But what are the reasons why $\trs{\pfone}$
can be partial? Clearly, it is not defined on any state $\stone=(\typeone,\pconone,\stckone,\regone)$
where $\typeone$ occurs in the conclusion of $\pfone$: $\pconone$ tells us that the next
state should be ``below $\stone$'', but there's nothing below the conclusion of $\pfone$.
For the same reasons, no state are mapped to a quadruple in the form
$(\typeone,\nconone,\stckone,\regone)$. This motivates the following definition:
the set $\ists{\pfone}$ of \emph{initial states} for a proof $\pfone$
consists of the states in $\sts{\pfone}$ in the form $(\typeone,\nconone,\emstck,\regone)$, where $\typeone$ is the
conclusion of $\pfone$. Analogously, \emph{final states} are
those in the form $(\typeone,\pconone,\emstck,\regone)$ and are the elements of $\fsts{\pfone}$.
The \emph{semantics} of $\pfone$ is the partial function 
$$
\sem{\pfone}:\ists{\pfone}\rightharpoonup\fsts{\pfone}
$$ 
defined by stipulating that $\sem{\pfone}(\stone)=\sttwo$ iff $\stone\trs{\pfone}^*\sttwo$.
One can prove that if we start from an initial state, we are guaranteed to reach a final state:
\begin{proposition}\label{prop:total}
For every $\pfone$, $\sem{\pfone}$ is total.
\end{proposition}
\begin{proof}
A state is said to be have a \emph{legal stack} if its third component is coherent with
the box-depth of its first component in the proof $\pfone$.
On the one hand, any state $\stone$ reachable from an initial state has the property of
having a legal stack, as can be easily proved by induction on the length of any chain of 
transitions leading any initial state to $\stone$. On the other hand, any state with
a legal stack is \emph{deadlock-free}, i.e. it is either final or such that $\stone\trs{\pfone}\sttwo$ for some $\sttwo$.
This means that starting from any initial state we can either reach a final state or
go on forever. But the latter is not possible, since $\trs{\pfone}$ is injective even
when restricted to the first three components of states, and moreover the set of states having
a legal stack (again, if we discard the quantum register) is finite.
\end{proof}
\begin{proposition}
If $\pfone\becomes\pftwo$, then $\sem{\pfone}=\sem{\pftwo}$.
\end{proposition}
\begin{proof}
It is an easy task to prove that each cut-elimination step
can possibly alter the underlying \QIAM, but in a way which cannot
be observed from the environment, i.e., by querying the machine
from an initial state.
\end{proof}
Lemma~\ref{lemma:cc} justifies the following definition: given a proof $\pfone$ with conclusion $\typeone$ and a negative context $\nconone$ for 
its $\typeone$, \emph{the semantics} of $\pfone$ relative to $\nconone$ is the
function
$$
\semp{\pfone}{\nconone}:\CC^{2^{\depth{\nconone}}}\longrightarrow\CC^{2^{\depth{\nconone}}}
$$
defined by stipulating that $\semp{\pfone}{\nconone}(\regone)=\regtwo$ iff 
$(\typeone,\nconone,\emstck,\regone)\trs{\pfone}^*(\typeone,\pconone,\emstck,\regtwo)$.
Noticeably:
\begin{theorem}\label{thm:soundness}
For every $\pfone$ and for every negative context $\nconone$ for the conclusion of $\pfone$,
$\semp{\pfone}{\nconone}$ is unitary. Moreover, a quantum circuit computing it can
be effectively extracted from $\nconone$.
\end{theorem}
\begin{proof}
An easy corollary of Lemma~\ref{lemma:cc}, Proposition~\ref{prop:total} and the fact
$\qiam{\pfone}$ is and effective and executable description of $\pfone$.
\end{proof}
%%%%%%%%%%%%%%%%%%%%%%%
\section{Conclusions}
%%%%%%%%%%%%%%%%%%%%%%%
Theorem~\ref{thm:soundness} establishes a sort of  \emph{soundness} result: any \QMLL\ proof
can be interpreted as a set of independent  unitary quantum circuits by way of a concrete GoI interpretation.
We already know (Section \ref{sec:QCencoding}) that 
any unitary quantum circuit $\qcone$ acting on $m$ qubits is captured by a \QMLL\ proof $\pfone_\qcone$ with
conclusion $\jd{\qqb{\bt{\atomone}}{m},\qq{\atomone}{m}}$. Hence,  \QMLL\ is somehow a \emph{complete} system for 
unitary quantum circuits.
 Some observations  are now in order:
\begin{varitemize}
\item
  The encoding is correct: the (unique!) circuit obtained
  from $\sem{\pfone_\qcone}$ by way of Theorem~\ref{thm:soundness} is $\qcone$.
\item
  The unitary operators in $\pfone_\qcone$ are exactly the quantum gates
  in the circuit $\qcone$. If we apply cut-elimination to $\pfone_\qcone$, however, some of those
  unitary operators are composed or tensorized. From a purely syntactical point of view this can be seen
  as a way to alter the quantum circuit underlying a proof, preserving equivalence.
\end{varitemize}
We also observe that the encoding does not make use of the multiplicative connectives at all. So, in a sense the
modal fragment of \QMLL\ is itself complete for quantum circuits. A further clarification
of the r\^ole of multiplicatives in \QMLL\ is a fascinating 
subject which we leave for future work.

\bibliographystyle{eptcs}
\bibliography{main}

\end{document}